\documentclass[12pt]{article}    

\usepackage[margin=0.75in]{geometry} 
\usepackage{amsmath,amssymb,amsthm}
\usepackage{color}

\newtheorem{theorem}{Theorem}[section]

\newtheorem{lemma}[theorem]{Lemma}

\theoremstyle{remark}

\newcommand{\be}{\begin{equation}}
\newcommand{\ee}{\end{equation}}

\newcommand{\bea}{\begin{eqnarray}}
\newcommand{\eea}{\end{eqnarray}}

\numberwithin{equation}{section}
\begin{document}

\title{Gap Probability Distribution of the Jacobi Unitary Ensemble: An Elementary Treatment, from Finite $n$ to Double Scaling}
\author{Chao Min\thanks{School of Mathematical Sciences, Huaqiao University, Quanzhou 362021, China; email: chaomin@hqu.edu.cn}\: and Yang Chen\thanks{Department of Mathematics, Faculty of Science and Technology, University of Macau, Macau, China;
email: yangbrookchen@yahoo.co.uk}}


\date{\today}
\maketitle
\begin{abstract}
In this paper, we study the gap probability problem of the (symmetric) Jacobi unitary ensemble of Hermitian random matrices, namely the probability that the interval $(-a,a)\:(0<a<1)$ is free of eigenvalues. Using the ladder operator technique for orthogonal polynomials and the associated supplementary conditions, we derive three quantities instrumental in the gap probability, denoted by $H_{n}(a)$, $R_{n}(a)$ and $r_{n}(a)$. We find that each one satisfies a second order differential equation. We show that after a double scaling, the large second order differential equation in the variable $a$ with $n$ as parameter satisfied by $H_{n}(a)$, can be reduced to the Jimbo-Miwa-Okamoto $\sigma$ form of the Painlev\'{e} V equation.
\end{abstract}

\section{Introduction}
In the work of Jimbo, Miwa, Mori and Sato \cite{Jimbo1980}, they studied the problem of the density matrix of an impenetrable Bose gas. They showed that the probability of observing a gap in the ground state can be related to a particular Painlev\'{e} V. The distribution of the largest eigenvalues of finite $n$ ensembles of random matrices generated by Gaussian, Laguerre and Jacobi weights for $\beta=1,2,4$ (here $\beta$ is the ``symmetry" parameter), has been systematically investigated by Adler and Van Moerbeke \cite{Adler}, via the vertex operators. Tracy and Widom \cite{Tracy} and Cao, Chen and Griffin \cite{Cao} studied the symmetric interval case of the gap probability problem for the Gaussian unitary ensemble, namely, the probability that the interval $J:=(-a,a)$ is free of eigenvalues.

In this paper, we consider the symmetric case for the Jacobi unitary ensemble generated by the weight $(1-x)^{\alpha}(1+x)^{\alpha}$.  The Jacobi unitary ensemble has various applications in particle physics, multivariate statistics and the symmetric spaces $\mathrm{O}^{+}(n)$ and $\mathrm{Sp}(n)/\mathrm{U}(n)$ \cite{Forrester2010}. This symmetric Jacobi unitary ensemble has also been studied by Witte, Forrester and Cosgrove \cite{Witte}, and Forrester \cite{Forrester2006}. Note that since the exponents in the weight are the same, this leads an even weight function, therefore we are in the similar situation as the Gaussian weight $\mathrm{e}^{-x^{2}}$. However,
 the derivation is much more complicated than the Gaussian case, since the support of $(1-x^{2})^{\alpha}$ is on $(-1,1)$. We mainly use the ladder operator approach to solve the problem, this theory has been widely applied to random matrix theory, the reader may refer to \cite{Basor2009, Basor2012, Cao, Chen1997L, Chen2004, Chen2012, Chen2010, Dai} for further information. For the Jacobi unitary ensemble and Painlev\'{e} VI equation, please refer to \cite{Chen2010, Haine, Dai}.

An elementary method to deal with the gap probability is to write it as a Hankel determinant, or a determinant of a moments matrix where the moments are generated by classical weight function, multiplied by one minus the characteristic function of an interval $J$ (see the formula (\ref{pan}) below). It is clear from such determinant representations that one is led to the study of the unconventional polynomials orthogonal with respect to the classical weight multiplied
by $1-\chi_{J}(x)$, that is $(1-x^{2})^{\alpha}\left(1-\chi_{(-a,a)}(x)\right)$. It is a well-known fact that Hankel determinants can be expressed as the product of the square of the $L^{2}$ norms of orthogonal polynomials. Based on the ladder operators adapted to these orthogonal polynomials, and from the associated supplementary conditions, a series difference and differential equations can be derived to ultimately give a description of the gap probability.

Let $P_{n}(x,a)$ be the monic orthogonal polynomials of degree $n$ with respect to the even weight function $w(x,a)$, namely,
$$
\int_{-1}^{1}P_{j}(x,a)P_{k}(x,a)w(x,a)dx=h_{j}(a)\delta_{jk},\;\;j,k=0,1,2,\ldots,
$$
where
$$
w(x,a):=w_{0}(x)\left(1-\chi_{(-a,a)}(x)\right),\;\;w_{0}(x):=\mathrm{e}^{-\mathrm{v}_{0}(x)}=(1-x^{2})^{\alpha},\;\;\alpha>0,
$$
and $\chi_{(-a,a)}(x)$ is the indicator function of the interval $(-a,a)$.
\\
Since $w(x,a)$ is even, $P_{n}(x,a)$ can be normalized as \cite{Chihara}
$$
P_{n}(x,a)=x^{n}+\mathrm{p}(n,a)x^{n-2}+\cdots+P_{n}(0,a).
$$
We shall see that $\mathrm{p}(n,a)$, the coefficient of $x^{n-2}$, will play an important role in the following discussions.

The three terms recurrence relation reduces to this special form since the weight is even \cite{Szego},
\be
xP_{n}(x,a)=P_{n+1}(x,a)+\beta_{n}(a)P_{n-1}(x,a)\label{re}
\ee
to be supplemented by the initial conditions
$$
P_{0}(x,a)=1,\;\;\beta_{0}(a)P_{-1}(x,a)=0.
$$
An easy computation shows that
$$
\beta_{n}(a)=\mathrm{p}(n,a)-\mathrm{p}(n+1,a)=\frac{h_{n}(a)}{h_{n-1}(a)}.
$$
Standard consideration \cite{Mehta} shows that the probability that $(-a,a)$ is free of eigenvalues in the Jacobi unitary ensemble is given by
\be\label{pan}
\mathbb{P}(a,n)=\frac{\det\left(\int_{-1}^{1}x^{i+j}w(x,a)dx\right)_{i,j=0}^{n-1}}{\det\left(\int_{-1}^{1}x^{i+j}w_{0}(x)dx\right)_{i,j=0}^{n-1}}
=\prod_{j=0}^{n-1}\frac{h_{j}(a)}{h_{j}(0)}.
\ee

The paper is organized as follows. In Sec. 2, we apply the ladder operators and associated supplementary conditions to the symmetric Jacobi weight, $(1-x^2)^{\alpha}$, on the set $J^{c}:=(-1,-a)\cup(a,1)$. In Sec. 3, we study the continuous evolution in $a$ and obtain the relations between $\beta_{n}(a)$ and the auxiliary quantity $r_{n}(a)$. In Sec. 4, we obtain the coupled Riccati equations satisfied by the auxiliary quantities $R_{n}(a)$ and $r_{n}(a)$. In Sec. 5, we establish the differential equation of $H_{n}(a)$, related to $\mathbb{P}(a,n)$, and finally reduced to the Jimbo-Miwa-Okamoto $\sigma$ form of the Painlev\'{e} V equation after a suitable double scaling, first obtained by Jimbo, Miwa, Mori and Sato \cite{Jimbo1980} in the impenetrable boson problem, where they studied the Fredholm determinant $\det\left(\mathbb{I}-K_{(-a,a)}\right)$ and $K_{(-a,a)}$ is the integral operator with kernel $K(x,y)=\frac{\sin(x-y)}{x-y}$ defined on the interval $(-a,a)$.

\section{Ladder Operators and Supplementary Conditions}
For convenience, we do not write down the $a$ dependence in $P_{n}(x)$, $w(x)$, $h_{n}$ and $\beta_{n}$ unless it is needed.
\begin{theorem}\label{th1}
The monic orthogonal polynomials with respect to the weight $w(x)$ on $[-1,1]$ satisfy the following differential recurrence relation
$$
P_{n}'(z)=\beta_{n}A_{n}(z)P_{n-1}(z)-B_{n}(z)P_{n}(z),
$$
where
$$
A_{n}(z)=\frac{a\:R_{n}(a)}{z^{2}-a^{2}}+\frac{\tilde{R}_{n}(a)}{z^{2}-1},
$$
$$
B_{n}(z)=\frac{z\:r_{n}(a)}{z^{2}-a^{2}}+\frac{z\:\tilde{r}_{n}(a)}{z^{2}-1},
$$
and
$$
R_{n}(a)=\frac{2P_{n}^{2}(a,a)(1-a^{2})^{\alpha}}{h_{n}(a)},
$$
$$
r_{n}(a)=\frac{2P_{n}(a,a)P_{n-1}(a,a)(1-a^{2})^{\alpha}}{h_{n-1}},
$$
$$
\tilde{R}_{n}(a)=\frac{2\alpha}{h_{n}(a)}\int_{-1}^{1}\frac{P_{n}^{2}(x,a)w(x,a)}{x^{2}-1}dx,
$$
$$
\tilde{r}_{n}(a)=\frac{2\alpha}{h_{n-1}}\int_{-1}^{1}\frac{xP_{n}(x,a)P_{n-1}(x,a)w(x,a)}{x^{2}-1}dx.
$$
\end{theorem}
\noindent{\bf Remark.}
The quantities $\tilde{R}_n(a)$ and $\tilde{r}_n(a)$ are a new features of this problem, not seen in the case where $w_{0}(x)=\mathrm{e}^{-x^2}$. Ultimately, $\tilde{R}_{n}(a)$ and $\tilde{r}_{n}(a)$ may be expressed in terms of $R_{n}(a)$ and $r_{n}(a)$, respectively, see (\ref{aux1}) and (\ref{aux2}).
\begin{proof}
We start from
$$
P_{n}'(z)=\sum_{k=0}^{n-1}c_{nk}P_{k}(z),
$$
and orthogonality implies,
$$
c_{nk}=\frac{1}{h_{k}}\int_{-1}^{1}P_{n}'(x)P_{k}(x)w(x)dx.
$$
Through integration by parts, we have
\bea
c_{nk}
&=&-\frac{1}{h_{k}}\int_{-1}^{1}P_{n}(x)P_{k}(x)w'(x)dx\nonumber\\
&=&-\frac{1}{h_{k}}\int_{-1}^{1}P_{n}(x)P_{k}(x)w_{0}'(x)\left(1-\chi_{(-a,a)}(x)\right)dx+\frac{1}{h_{k}}\int_{-1}^{1}P_{n}(x)P_{k}(x)w_{0}(x)\chi_{(-a,a)}'(x)dx
\nonumber\\
&=&-\frac{1}{h_{k}}\int_{-1}^{1}P_{n}(x)P_{k}(x)\frac{w_{0}'(x)}{w_{0}(x)}w(x)dx+\frac{1}{h_{k}}\int_{-1}^{1}P_{n}(x)P_{k}(x)w_{0}(x)\left[\theta(x+a)-\theta(x-a)\right]'dx\nonumber\\
&=&\frac{1}{h_{k}}\int_{-1}^{1}P_{n}(x)P_{k}(x)\mathrm{v}_{0}'(x)w(x)dx+\frac{1}{h_{k}}\int_{-1}^{1}P_{n}(x)P_{k}(x)w_{0}(x)
\left[\delta(x+a)-\delta(x-a)\right]dx\nonumber\\
&=&\frac{1}{h_{k}}\int_{-1}^{1}P_{n}(x)P_{k}(x)\left[\mathrm{v}_{0}'(x)-\mathrm{v}_{0}'(z)\right]w(x)dx-\frac{1}{h_{k}}\left[P_{n}(a)P_{k}(a)(1-a^{2})^{\alpha}-P_{n}(-a)P_{k}(-a)(1-a^{2})^{\alpha}\right],\nonumber
\eea
where $\theta(x)$ is the Heaviside step function, i.e., $\theta(x)$ is 1 for $x>0$ and 0 otherwise.
\\
It follows that
\bea
P_{n}'(z)
&=&\int_{-1}^{1}\sum_{k=0}^{n-1}\frac{P_{k}(z)P_{k}(x)}{h_{k}}P_{n}(x)\left[\mathrm{v}_{0}'(x)-\mathrm{v}_{0}'(z)\right]w(x)dx\nonumber\\
&-&\sum_{k=0}^{n-1}\frac{P_{k}(z)P_{k}(a)}{h_{k}}P_{n}(a)(1-a^{2})^{\alpha}+\sum_{k=0}^{n-1}\frac{P_{k}(z)P_{k}(-a)}{h_{k}}P_{n}(-a)(1-a^{2})^{\alpha}.
\nonumber
\eea
Using the Christoffel-Darboux formula, namely,
$$
\sum_{k=0}^{n-1}\frac{P_k(x)P_k(y)}{h_k}=\frac{P_n(x)P_{n-1}(y)-P_n(y)P_{n-1}(x)}{h_{n-1}(x-y)},
$$
we find,
\bea
P_{n}'(z)
&=&-\frac{1}{h_{n-1}}\int_{-1}^{1}\frac{\mathrm{v}_{0}'(z)-\mathrm{v}_{0}'(x)}{z-x}\left[P_{n}(z)P_{n-1}(x)-P_{n}(x)P_{n-1}(z)\right]P_{n}(x)w(x)dx
\nonumber\\
&-&\frac{P_{n}(z)P_{n-1}(a)-P_{n}(a)P_{n-1}(z)}{h_{n-1}(z-a)}P_{n}(a)(1-a^{2})^{\alpha}\nonumber\\
&+&\frac{P_{n}(z)P_{n-1}(-a)-P_{n}(-a)P_{n-1}(z)}{h_{n-1}(z+a)}P_{n}(-a)(1-a^{2})^{\alpha}.\label{pnp}
\eea
Since the weight $w(x)$ is even, we have,
$$
P_{n}(-a,a)=(-1)^{n}P_{n}(a,a).
$$
Observe that
$$
\frac{\mathrm{v}_{0}'(z)-\mathrm{v}_{0}'(x)}{z-x}=\frac{2\alpha(1+xz)}{(x^{2}-1)(z^{2}-1)}.
$$
Then (\ref{pnp}) becomes
\bea
P_{n}'(z)
&=&\left[\frac{1}{h_{n-1}(z^{2}-1)}\int_{-1}^{1}\frac{2\alpha(1+xz)}{x^{2}-1}P_{n}^{2}(x)w(x)dx
+\frac{2aP_{n}^{2}(a)(1-a^{2})^{\alpha}}{h_{n-1}(z^{2}-a^{2})}\right]P_{n-1}(z)\nonumber\\
&-&\left[\frac{1}{h_{n-1}(z^{2}-1)}\int_{-1}^{1}\frac{2\alpha(1+xz)}{x^{2}-1}P_{n}(x)P_{n-1}(x)w(x)dx
+\frac{2zP_{n}(a)P_{n-1}(a)(1-a^{2})^{\alpha}}{h_{n-1}(z^{2}-a^{2})}\right]P_{n}(z).\nonumber
\eea
It is easy to see from the parity of the integrand,
$$
\int_{-1}^{1}\frac{2\alpha(1+xz)}{x^{2}-1}P_{n}^{2}(x)w(x)dx=2\alpha\int_{-1}^{1}\frac{P_{n}^{2}(x)w(x)}{x^{2}-1}dx
$$
and
$$
\int_{-1}^{1}\frac{2\alpha(1+xz)}{x^{2}-1}P_{n}(x)P_{n-1}(x)w(x)dx=2\alpha z\int_{-1}^{1}\frac{xP_{n}(x)P_{n-1}(x)w(x)}{x^{2}-1}dx.
$$
With the aid of the formula $\beta_{n}=\frac{h_{n}}{h_{n-1}}$, it follows that
$$
P_{n}'(z)=\beta_{n}A_{n}(z)P_{n-1}(z)-B_{n}(z)P_{n}(z).
$$
The proof is complete.
\end{proof}

\begin{lemma}
The functions $A_{n}(z)$ and $B_{n}(z)$ satisfy the following identities:
\be
B_{n+1}(z)+B_{n}(z)=z A_{n}(z)-\mathrm{v}_{0}'(z), \tag{$S_{1}$}
\ee
\be
1+z(B_{n+1}(z)-B_{n}(z))=\beta_{n+1}A_{n+1}(z)-\beta_{n}A_{n-1}(z), \tag{$S_{2}$}
\ee
\be
B_{n}^{2}(z)+\mathrm{v}_{0}'(z)B_{n}(z)+\sum_{j=0}^{n-1}A_{j}(z)=\beta_{n}A_{n}(z)A_{n-1}(z). \tag{$S_{2}'$}
\ee
\end{lemma}
The three identities ($S_{1}$), ($S_{2}$) and the sum rule ($S_{2}'$) valid for $z\in \mathbb{C}\bigcup\{\infty\},$ a version of which, can be found in Magnus \cite{Magnus}, and were fundamental in the work appearing in \cite{Basor2009, Basor2012, Cao, Chen1997L, Chen2004, Chen2012, Chen2010, Dai}.

\noindent{\bf Remark.} $\mathrm{v}_{0}'(z)$ is the derivative of the ``smooth part'' of the weight appeared in ($S_{1}$) and ($S_{2}'$).

\begin{lemma}\label{le1}
\be\label{aux1}
\tilde{R}_{n}(a)=-\left(aR_{n}(a)+2n+2\alpha+1\right),
\ee
\be\label{aux2}
\tilde{r}_{n}(a)=-\left(r_{n}(a)+n\right).
\ee
\end{lemma}
\begin{proof}
We first establish the second equation above. From the definition of $\tilde{r}_{n}(a)$,
\bea
\tilde{r}_{n}(a)
&=&-\frac{2\alpha}{h_{n-1}}\int_{-1}^{1}x P_{n}(x)P_{n-1}(x)(1-x^{2})^{\alpha-1}\left(1-\chi_{(-a,a)}(x)\right)dx\nonumber\\
&=&\frac{1}{h_{n-1}}\int_{-1}^{1}P_{n}(x)P_{n-1}(x)\left(1-\chi_{(-a,a)}(x)\right)d(1-x^{2})^{\alpha}\nonumber\\
&=&-\frac{1}{h_{n-1}}\bigg[\int_{-1}^{1}P_{n}'(x)P_{n-1}(x)w(x)dx+\int_{-1}^{1}P_{n}(x)P_{n-1}'(x)w(x)dx\nonumber\\
&-&\int_{-1}^{1}P_{n}(x)P_{n-1}(x)(1-x^{2})^{\alpha}\chi_{(-a,a)}'(x)dx\bigg].\nonumber
\eea
Since
$$
P_{n}'(x)=nP_{n-1}(x)+\mathrm{a}\; \mathrm{polynomial}\; \mathrm{of}\; \mathrm{lower}\; \mathrm{degree},
$$
it follows that
$$
\int_{-1}^{1}P_{n}'(x)P_{n-1}(x)w(x)dx=nh_{n-1}.
$$
Hence,
\bea
\tilde{r}_{n}(a)
&=&-\frac{1}{h_{n-1}}\left[nh_{n-1}+2P_{n}(a)P_{n-1}(a)(1-a^{2})^{\alpha}\right]\nonumber\\
&=&-\left(r_{n}(a)+n\right).\label{trn}
\eea
Now we turn to prove the first equality. It turns out that it is better not to ``integrate by parts'' in $\tilde{R}_{n}(a)$. Instead, from ($S_{1}$), we obtain
$$
\frac{r_{n+1}(a)+r_{n}(a)-aR_{n}(a)}{z^{2}-a^{2}}=\frac{\tilde{R}_{n}(a)-\tilde{r}_{n+1}(a)-\tilde{r}_{n}(a)+2\alpha}{z^{2}-1}.
$$
Equating the residues at the simple pole $a$ and 1, we find
\be\label{s11}
r_{n+1}(a)+r_{n}(a)=aR_{n}(a)
\ee
and
\be\label{s12}
\tilde{R}_{n}(a)=\tilde{r}_{n+1}(a)+\tilde{r}_{n}(a)-2\alpha,
\ee
respectively.

Using (\ref{trn}) and (\ref{s11}), (\ref{s12}) becomes
\bea
\tilde{R}_{n}(a)&=&-r_{n+1}(a)-r_{n}(a)-2n-2\alpha-1\nonumber\\
&=&-\left(aR_{n}(a)+2n+2\alpha+1\right).\nonumber
\eea
This completes the proof.
\end{proof}
Using Theorem \ref{th1} and Lemma \ref{le1}, ($S_{2}$) becomes
\bea
&&\frac{z^{2}(r_{n+1}(a)-r_{n}(a))-a(\beta_{n+1}R_{n+1}(a)-\beta_{n}R_{n-1}(a))}{z^{2}-a^{2}}+1\nonumber\\
&=&\frac{z^{2}(r_{n+1}(a)-r_{n}(a)+1)-a(\beta_{n+1}R_{n+1}(a)-\beta_{n}R_{n-1}(a))-[(2n+2\alpha+3)\beta_{n+1}-(2n+2\alpha-1)\beta_{n}]}{z^{2}-1}.\nonumber
\eea
Similarly, equating the residues at the simple pole $a$ and 1, we obtain
\be\label{s21}
a(r_{n+1}(a)-r_{n}(a))=\beta_{n+1}R_{n+1}(a)-\beta_{n}R_{n-1}(a)
\ee
and
\be\label{s22}
r_{n+1}(a)-r_{n}(a)+1=a(\beta_{n+1}R_{n+1}(a)-\beta_{n}R_{n-1}(a))+(2n+2\alpha+3)\beta_{n+1}-(2n+2\alpha-1)\beta_{n},
\ee
respectively.

Substituting (\ref{s21}) to (\ref{s22}), we find,
$$
r_{n+1}(a)-r_{n}(a)+1=a^{2}(r_{n+1}(a)-r_{n}(a))+(2n+2\alpha+3)\beta_{n+1}-(2n+2\alpha-1)\beta_{n}.
$$
Hence,
$$
(a^{2}-1)(r_{n+1}(a)-r_{n}(a))+(2n+2\alpha+3)\beta_{n+1}-(2n+2\alpha-1)\beta_{n}=1.
$$
A telescopic sum gives
\bea
(a^{2}-1)r_{n}(a)
&=&-2\sum_{j=1}^{n-1}\beta_{j}-(2n+2\alpha+1)\beta_{n}+n\nonumber\\
&=&2\mathrm{p}(n,a)-(2n+2\alpha+1)\beta_{n}+n.\nonumber
\eea
Hence, we have,
\be
\mathrm{p}(n,a)=\frac{1}{2}\left[(a^{2}-1)r_{n}(a)+(2n+2\alpha+1)\beta_{n}-n\right].\label{pna}
\ee

From ($S_{2}'$), we have
\bea
&&\frac{z^{2}r_{n}^{2}(a)}{(z^{2}-a^{2})^{2}}-\frac{2z^{2}[r_{n}^{2}(a)+(n+\alpha)r_{n}(a)]}{(z^{2}-a^{2})(z^{2}-1)}
+\frac{z^{2}\left[r_{n}^{2}(a)+2(n+\alpha)r_{n}(a)+n^{2}+2n\alpha\right]}{(z^{2}-1)^{2}}\nonumber\\
&+&\frac{a\sum_{j=0}^{n-1}R_{j}(a)}{z^{2}-a^{2}}
-\frac{a\sum_{j=0}^{n-1}R_{j}(a)+n^{2}+2n\alpha}{z^{2}-1}\nonumber\\
&=&\beta_{n}\bigg[\frac{a^{2}R_{n}(a)R_{n-1}(a)}{(z^{2}-a^{2})^{2}}-\frac{2a^{2}R_{n}(a)R_{n-1}(a)+(2n+2\alpha+1)aR_{n-1}(a)
+(2n+2\alpha-1)aR_{n}(a)}{(z^{2}-a^{2})(z^{2}-1)}\nonumber\\
&+&\frac{(aR_{n}(a)+2n+2\alpha+1)(aR_{n-1}(a)+2n+2\alpha-1)}{(z^{2}-1)^{2}}\bigg].\label{s2p}
\eea
Multiplying both sides of (\ref{s2p}) by $(z^2-a^2)^2$ and letting $z\rightarrow a$, we obtain
\be\label{rn}
r_{n}^{2}(a)=\beta_{n}R_{n}(a)R_{n-1}(a)
\ee
Similarly, multiplying both sides of (\ref{s2p}) by $(z^2-1)^2$ and letting $z\rightarrow 1$, we find
\be\label{rn1}
r_{n}^{2}(a)+2(n+\alpha)r_{n}(a)+n^{2}+2n\alpha=\beta_{n}(aR_{n}(a)+2n+2\alpha+1)(aR_{n-1}(a)+2n+2\alpha-1).
\ee

It follows that
\bea
&&r_{n}^{2}(a)+2(n+\alpha)r_{n}(a)+n^{2}+2n\alpha\nonumber\\
&=&a^{2}\beta_{n}R_{n}(a)R_{n-1}(a)+(2n+2\alpha-1)a\beta_{n}R_{n}(a)+(2n+2\alpha+1)a\beta_{n}R_{n-1}(a)\nonumber\\
&+&(2n+2\alpha+1)(2n+2\alpha-1)\beta_{n}\nonumber\\
&=&a^{2}r_{n}^{2}(a)+(2n+2\alpha-1)a\beta_{n}R_{n}(a)+(2n+2\alpha+1)a\beta_{n}R_{n-1}(a)+(2n+2\alpha+1)(2n+2\alpha-1)\beta_{n},\nonumber
\eea
where we have made use of (\ref{rn}) in the last step. Then
\bea
&&(1-a^{2})r_{n}^{2}(a)+2(n+\alpha)r_{n}(a)+n^{2}+2n\alpha\nonumber\\
&=&(2n+2\alpha-1)a\beta_{n}R_{n}(a)+(2n+2\alpha+1)a\beta_{n}R_{n-1}(a)+(2n+2\alpha+1)(2n+2\alpha-1)\beta_{n}.\nonumber\\
\label{eq3}
\eea
With the aid of (\ref{rn}) and (\ref{rn1}), (\ref{s2p}) becomes
\bea
&&\frac{r_{n}^{2}(a)+a\sum_{j=0}^{n-1}R_{j}(a)}{z^{2}-a^{2}}-\frac{2z^{2}[r_{n}^{2}(a)+(n+\alpha)r_{n}(a)]}{(z^{2}-a^{2})(z^{2}-1)}
+\frac{r_{n}^{2}(a)+2(n+\alpha)r_{n}(a)-a\sum_{j=0}^{n-1}R_{j}(a)}{z^{2}-1}\nonumber\\
&=&-\frac{2a^{2}r_{n}^{2}(a)+(2n+2\alpha+1)a\beta_{n}R_{n-1}(a)+(2n+2\alpha-1)a\beta_{n}R_{n}(a)}{(z^{2}-a^{2})(z^{2}-1)}.\label{s2p1}
\eea
Equating the residues at the simple pole $a$ from (\ref{s2p1}), we find
\be
(1-a^{2})r_{n}^{2}(a)+2(n+\alpha)a^{2}r_{n}(a)+a(1-a^{2})\sum_{j=0}^{n-1}R_{j}(a)=(2n+2\alpha-1)a\beta_{n}R_{n}(a)+(2n+2\alpha+1)a\beta_{n}R_{n-1}(a).\label{eq4}
\ee
Subtract (\ref{eq4}) from (\ref{eq3}), we get
\be
a(a^{2}-1)\sum_{j=0}^{n-1}R_{j}(a)-2(n+\alpha)(a^{2}-1)r_{n}(a)-(2n+2\alpha+1)(2n+2\alpha-1)\beta_{n}+n^{2}+2n\alpha=0.\label{eq5}
\ee
\noindent{\bf Remark.}
We have the same result (\ref{eq4}) if we equate the residues at the simple pole 1 from (\ref{s2p1}).

\section{Evolution in $a$}
We start by taking the derivative with respect to $a$ in the following equation,
$$
\int_{-1}^{1}P_{n}^{2}(x)(1-x^{2})^{\alpha}\left(1-\chi_{(-a,a)}(x)\right)dx=h_{n}(a),\;\;n=0,1,2,\ldots,
$$
we obtain
$$
h_{n}'(a)=-2(1-a^{2})^{\alpha}P_{n}^{2}(a).
$$
It follows that
$$
\frac{d}{da}\ln h_{n}(a)=-R_{n}(a)
$$
and
$$
\frac{d}{da}\ln\beta_{n}(a)=\frac{d}{da}\ln h_{n}(a)-\frac{d}{da}\ln h_{n-1}(a)=R_{n-1}(a)-R_{n}(a).
$$
Hence, we have
$$
\beta_{n}'(a)=\beta_{n}R_{n-1}(a)-\beta_{n}R_{n}(a)
$$
or
\be
\beta_{n}R_{n-1}(a)=\beta_{n}'(a)+\beta_{n}R_{n}(a).\label{beta}
\ee
Taking a derivative with respect to $a$ in (\ref{pan}), we see that,
$$
\frac{d}{da}\ln\mathbb{P}(a,n)=\frac{d}{da}\ln\prod_{j=0}^{n-1}h_{j}(a)=-\sum_{j=0}^{n-1}R_{j}(a).
$$

On the other hand, taking the derivative with respect to $a$ in the equation
$$
\int_{-1}^{1}P_{n}(x)P_{n-2}(x)(1-x^{2})^{\alpha}\left(1-\chi_{(-a,a)}(x)\right)dx=0,\;\;n=0,1,2,\ldots
$$
produces
$$
\frac{d\mathrm{p}(n,a)}{da}=\frac{2(1-a^{2})^{\alpha}P_{n}(a)P_{n-2}(a)}{h_{n-2}}.
$$
According to the recurrence relation (\ref{re}), and note that $\beta_{n}=\frac{h_{n}}{h_{n-1}}$, we have
$$
\frac{P_{n-2}(a)}{h_{n-2}}=\frac{aP_{n-1}(a)}{h_{n-1}}-\frac{P_{n}(a)}{h_{n-1}}.
$$
It follows that
\be
\frac{d\mathrm{p}(n,a)}{da}=a r_{n}(a)-\beta_{n}R_{n}(a).\label{pna1}
\ee
\begin{lemma}
$\beta_{n}(a)$ and $r_{n}(a)$ satisfy the following differential equations:
\be\label{rnbn}
(2n+2\alpha+1)(2n+2\alpha-1)(\beta_{n}'(a))^{2}+4(n+\alpha)(a^{2}-1)\beta_{n}'(a)r_{n}'(a)+(a^{2}-1)^{2}(r_{n}'(a))^{2}=4\beta_{n}r_{n}^{2}(a).
\ee
\be\label{equ3}
(2n+2\alpha+1)(2n+2\alpha-1)(\beta_{n}-a\beta_{n}'(a))=(1-a^{2})r_{n}^{2}(a)+2(n+\alpha)r_{n}(a)-2(n+\alpha)a(1-a^{2})r_{n}'(a)+n^{2}+2n\alpha.
\ee
\end{lemma}
\begin{proof}
We begin by substituting (\ref{pna}) into (\ref{pna1}), to find,
\be\label{rnp}
(1-a^{2})r_{n}'(a)=2\beta_{n}R_{n}(a)+(2n+2\alpha+1)\beta_{n}'(a).
\ee
Then
\be
\beta_{n}R_{n}(a)=\frac{1}{2}\left[(1-a^{2})r_{n}'(a)-(2n+2\alpha+1)\beta_{n}'(a)\right].\label{br}
\ee
Combining (\ref{rn}) with (\ref{beta}) gives,
\be
(\beta_{n}'(a)+\beta_{n}R_{n}(a))\beta_{n}R_{n}(a)=\beta_{n}r_{n}^{2}(a).\label{bn}
\ee
Replacing $\beta_{n}R_{n}(a)$ by (\ref{br}), (\ref{bn}) becomes,
$$
(2n+2\alpha+1)(2n+2\alpha-1)(\beta_{n}'(a))^{2}+4(n+\alpha)(a^{2}-1)\beta_{n}'(a)r_{n}'(a)+(a^{2}-1)^{2}(r_{n}'(a))^{2}=4\beta_{n}r_{n}^{2}(a).
$$

Now we come to prove the second equation. Substituting (\ref{beta}) into (\ref{eq3}) to eliminate $\beta_{n}R_{n-1}(a)$, we find
\bea
&&(4n+4\alpha)a\beta_{n}R_{n}(a)\nonumber\\
&=&(1-a^{2})r_{n}^{2}(a)+2(n+\alpha)r_{n}(a)+n^{2}+2n\alpha-(2n+2\alpha+1)(2n+2\alpha-1)\beta_{n}-(2n+2\alpha+1)a\beta_{n}'(a).\nonumber\\
\label{br2}
\eea
Replacing $\beta_{n}R_{n}(a)$ by (\ref{br}), (\ref{br2}) becomes,
$$
(2n+2\alpha+1)(2n+2\alpha-1)(\beta_{n}-a\beta_{n}'(a))=(1-a^{2})r_{n}^{2}(a)+2(n+\alpha)r_{n}(a)-2(n+\alpha)a(1-a^{2})r_{n}'(a)+n^{2}+2n\alpha.
$$
The lemma is established.
\end{proof}

\section{Coupled Riccati Equations Satisfied by $R_{n}(a)$ and $r_{n}(a)$}
In this section, we are able to obtain the differential equations on $R_{n}(a)$ and $r_{n}(a)$ respectively. From (\ref{rn}) and (\ref{beta}), we have
\be\label{bep}
\beta_{n}'(a)=\frac{r_{n}^2(a)}{R_{n}(a)}-\beta_{n}R_{n}(a).
\ee
Plugging (\ref{bep}) into (\ref{rnp}), we find,
\be\label{rnp1}
r_{n}'(a)=\frac{(2n+2\alpha+1)r_{n}^2(a)}{(1-a^2)R_{n}(a)}-\frac{2n+2\alpha-1}{1-a^2}\beta_{n}R_{n}(a).
\ee
Combining (\ref{eq3}) and (\ref{rn}), we have
\bea
&&(1-a^{2})r_{n}^{2}(a)+2(n+\alpha)r_{n}(a)+n^{2}+2n\alpha\nonumber\\
&=&(2n+2\alpha-1)a\beta_{n}R_{n}(a)+(2n+2\alpha+1)a \frac{r_{n}^2(a)}{R_{n}(a)}+(2n+2\alpha+1)(2n+2\alpha-1)\beta_{n}.\nonumber
\eea
It follows that
\be\label{beta1}
\beta_{n}(a)=\frac{[(1-a^{2})r_{n}^{2}(a)+2(n+\alpha)r_{n}(a)+n^{2}+2n\alpha]R_{n}(a)-(2n+2\alpha+1)a r_{n}^2(a)}{(2n+2\alpha-1)R_{n}(a)(a R_{n}(a)+2n+2\alpha+1)}.
\ee
Putting (\ref{beta1}) into (\ref{rnp1}), we find a Riccati equation on $r_{n}(a)$,
\bea\label{rnp2}
r_{n}'(a)&=&\frac{\left[(a^2-1)R_{n}^2(a)+2a(2n+2\alpha+1)R_{n}(a)+(2n+2\alpha+1)^2\right]r_{n}^2(a)}{(1-a^2)R_{n}(a)(aR_{n}(a)+2n+2\alpha+1)}\nonumber\\
&-&\frac{2(n+\alpha)R_{n}(a)r_{n}(a)+n(n+2\alpha)R_{n}(a)}{(1-a^2)(aR_{n}(a)+2n+2\alpha+1)}.
\eea
Taking (\ref{beta1}) into (\ref{bep}) to eliminate $\beta_{n}$ and $\beta_{n}'(a)$, and using (\ref{rnp2}) to replace $r_{n}'(a)$, we obtain a Riccati equation on $R_{n}(a)$,
\be\label{ri}
R_{n}'(a)=R_{n}^2(a)+\frac{2a^2(n+\alpha)-2(a^2-1)r_{n}(a)}{a(a^2-1)}R_{n}(a)-\frac{2(2n+2\alpha+1)}{a^2-1}r_{n}(a).
\ee
\noindent{\bf Remark.}
Similar coupled Riccati equations have appeared in the theory of Painlev\'{e} equations, see V. I. Gromak, I. Laine and S. Shimomura \cite{Gromak}. In partcular, see V. I. Gromak publications in \cite{Gromak}, page 287-288.

Eliminating $r_{n}(a)$ from (\ref{ri}) and (\ref{rnp2}), we obtain a second order differential equation satisfied by $R_{n}(a)$,
\newpage
\bea
&&2a(a^2-1)^2 R_{n}(a)(a R_{n}(a)+2n+2\alpha+1)\left[(a^2-1)R_{n}(a)+(2n+2\alpha+1)a\right]R_{n}''(a)\nonumber\\
&-&a(a^2-1)^2\left[3a(a^2-1)R_{n}^2(a)+2(2a^2-1)(2n+2\alpha+1)R_{n}(a)+(2n+2\alpha+1)^2a\right](R_{n}'(a))^2\nonumber\\
&+&2(a^2-1)(2n+2\alpha+1)R_{n}(a)\left[(2a^4-a^2+1)R_{n}(a)+2a^3(2n+2\alpha+1)\right]R_{n}'(a)-a^2(a^2-1)^3R_{n}^6(a)\nonumber\\
&-&2a(a^2-1)^2(2a^2-1)(2n+2\alpha+1)R_{n}^5(a)-(a^2-1)\Big[a^4(7+24\alpha+24n+24n^2+48n\alpha+24\alpha^2)\nonumber\\
&-&a^2(5+24n+24n^2+24\alpha+48n\alpha+20\alpha^2)+2+4\alpha+4n+4n^2+8n\alpha\Big]R_{n}^4(a)\nonumber\\
&-&4a(a^2-1)(2n+2\alpha+1)\left[a^2(2n+2\alpha+1)^2-2n-2n^2-2\alpha-4n\alpha\right]R_{n}^3(a)\nonumber\\
&-&4a^2(2n+2\alpha+1)^2\left[(n(n+1)+(2n+1)\alpha+\alpha^2)a^2-n(n+1)-(2n+1)\alpha\right]R_{n}^2(a)=0.\nonumber
\eea
Similarly, eliminating $R_{n}(a)$ from (\ref{rnp2}) and (\ref{ri}), we obtain a second order differential equation satisfied by $r_{n}(a)$,
\bea
&&a^2 (a^2-1)^4 (r_n''(a))^2+4 a^2 (a^2-1)^2  \left[a (a^2-1) r_n'(a)+4 r_n^3(a)+6 (\alpha +n) r_n^2(a)+2 n (2 \alpha +n) r_n(a)\right]r_n''(a)\nonumber\\
&-&4 (a^2-1)^2 \left[(a^2+1)^2 r_n^2(a)+2 (a^2+1) a^2 (\alpha +n) r_n(a)+a^4 (\alpha +n-1) (\alpha +n+1)\right] (r_n'(a))^2\nonumber\\
&+&16 a^3 (a^2-1) \left[2 r_n^3(a)+3 (\alpha +n) r_n^2(a)+n (2 \alpha +n) r_n(a)\right] r_n'(a)-16 (a^2-1)^2 r_n^6(a)\nonumber\\
&-&32 (2 a^4-3 a^2+1) (\alpha +n) r_n^5(a)-16 (a^2-1) \left[5 a^2 \alpha ^2+(6 a^2-1) (n^2+2\alpha  n)\right] r_n^4(a)\nonumber\\
&-&32 a^2 (\alpha +n) \left[a^2 \alpha ^2+(a^2-1) (2 n^2+4 \alpha  n)\right] r_n^3(a)-16 a^2 n (2 \alpha +n) \left[a^2 \alpha ^2+(a^2-1) (n^2+2 \alpha  n)\right] r_n^2(a)\nonumber\\
&=&0.\label{cha}
\eea
\noindent{\bf Remark.}
Equation (\ref{cha}) may be transformed to the fourth member of the Chazy II system \cite{Cosgrove, Chazy1909, Chazy1911}.

\section{$\sigma$ Form of the Painlev\'{e} V Equation}
We introduce the quantity $H_{n}(a)$, defined by
$$
H_{n}(a):=a(a^{2}-1)\frac{d}{da}\ln\mathbb{P}(a,n)=a(1-a^{2})\sum_{j=0}^{n-1}R_{j}(a).
$$
Then (\ref{eq5}) becomes,
$$
H_{n}(a)-2(n+\alpha)(1-a^{2})r_{n}(a)+(2n+2\alpha+1)(2n+2\alpha-1)\beta_{n}-n^{2}-2n\alpha=0.
$$
It follows that
\be
(2n+2\alpha+1)(2n+2\alpha-1)\beta_{n}=-H_{n}(a)+2(n+\alpha)(1-a^{2})r_{n}(a)+n^{2}+2n\alpha.\label{equ1}
\ee
Putting (\ref{equ1}) into (\ref{equ3}), eliminate $\beta_{n}$ and $\beta_{n}'(a)$, we obtain
\be\label{equ4}
a H_{n}'(a)-H_{n}(a)=(1-a^{2})r_{n}^{2}(a)-2(n+\alpha)a^{2}r_{n}(a).
\ee
Taking a derivative on both sides of (\ref{equ4}), we find
$$
r_{n}'(a)=\frac{a H_{n}''(a)+2a r_{n}^2(a)+4a(n+\alpha)r_{n}(a)}{2[(1-a^2)r_{n}(a)-a^2(n+\alpha)]}.
$$
Then
\bea
(r_{n}'(a))^2&=&\frac{[a H_{n}''(a)+2a r_{n}^2(a)+4a(n+\alpha)r_{n}(a)]^2}{4[(1-a^2)r_{n}(a)-a^2(n+\alpha)]^2}\nonumber\\
&=&\frac{[a H_{n}''(a)+2a r_{n}^2(a)+4a(n+\alpha)r_{n}(a)]^2}{4[(1-a^2)(a H_{n}'(a)-H_{n}(a))+(n+\alpha)^2a^4]},\label{rnp4}
\eea
where the second equality comes from (\ref{equ4}).

On the other hand, taking (\ref{equ1}) to (\ref{rnbn}), eliminate $\beta_{n}$ and $\beta_{n}'(a)$, we have
\bea
(a^2-1)^2 (r_{n}'(a))^2&=&8(a^2-1)(n+\alpha)r_{n}^3(a)+4[(4a^2-1)(n+\alpha)^2+\alpha^2+H_{n}(a)]r_{n}^2(a)\nonumber\\
&+&8a(n+\alpha)H_{n}'(a)r_{n}(a)+(H_{n}'(a))^2.\label{rnp5}
\eea
Substituting (\ref{rnp4}) into (\ref{rnp5}), and using (\ref{equ4}) to eliminate $r_{n}^2(a)$ and $r_{n}^3(a)$, to get an equation linear in $r_{n}(a)$. The solution is
\be\label{rna}
r_{n}(a)=
\frac{N}{D},
\ee
where $N$ and $D$ are exclusively expressed in terms of $H_n(a)$, $H_n'(a)$ and $H_n''(a)$,
\bea
N:
&=&a^2 (a^2-1)^3 (H_n''(a))^2+4 a^2 (a^2-1)^2\left[H_n(a)-a H_n'(a)\right] H_n''(a)+4 a (a^2-1)^2 (H_n'(a))^3\nonumber\\
&-&4 (a^2-1) \left[(5 a^2-1) H_n(a)+ (\alpha+n-1) (\alpha+n+1)a^4-4  n (n+2 \alpha)a^2\right](H_n'(a))^2\nonumber\\
&+&\Big\{32a (a^2-1) H_n^2(a)+ \left[8 a^5 (2 \alpha ^2+2 n^2+4 \alpha  n-1)-8 a^3 (4 n^2+8 \alpha  n-1)+32 a n (2 \alpha +n)\right]H_n(a)\nonumber\\
&-&16 a^3 (\alpha +n)^2 (a^2 n^2+2 \alpha  a^2 n+1)\Big\}H_n'(a)-16 (a^2-1) H_n^3(a)\nonumber\\
&-&\left[4 a^4 (4 \alpha ^2+4 n^2+8 \alpha  n-1)-4 a^2 (4 n^2+8 \alpha  n-1)+16 n (2 \alpha +n)\right]H_n^2(a)\nonumber\\
&+&16 a^2 (\alpha +n)^2(a^2 n^2+2 \alpha  a^2 n+1) H_n(a),\nonumber
\eea
\bea
D:
&=&8 (\alpha +n) \Big\{(a^2-1)^2 a^2 H_n''(a)-2a (a^2-1) \left[a^2 (2 n^2+4 \alpha  n+1)-2 H_n(a)\right] H_n'(a)-4 (a^2-1) H_n^2(a)\nonumber\\
&-&2 a^2  \left[(2 \alpha ^2-1) a^2+2 n^2+4 \alpha  n+1\right]H_n(a)+4 a^4 (\alpha +n)^2 (a^2 n^2+2 \alpha  a^2 n+1)\Big\}.\nonumber
\eea
Substituting (\ref{rna}) into (\ref{equ4}), we obtain a second order differential equation on $H_{n}(a)$,
\bea
&&\bigg\{a^2 (a^2-1)^4 (H_n''(a))^2-4 a^2 (a^2-1)^2  \Big[(a^3-a) H_n'(a)-(a^2-1) H_n(a)-2 a^2 (\alpha +n)^2\Big]H_n''(a)\nonumber\\
&+&4 a (a^2-1)^3 (H_n'(a))^3-4 (a^2-1)^2 \Big[a^4 (\alpha +n-1) (\alpha +n+1)-4 a^2 n (2 \alpha +n)\nonumber\\
&+&(5 a^2-1) H_n(a)\Big] (H_n'(a))^2+8 a (a^2-1) \Big[4 (a^2-1) H_n^2(a)\nonumber\\
&+&\left(a^4 (2 \alpha ^2+2 n^2+4 \alpha  n-1)+a^2(4 \alpha ^2+1) +4 n (2 \alpha +n)\right)H_n(a)\nonumber\\
&-&2 a^2 (\alpha +n)^2 (3 a^2 n^2+6 \alpha  a^2 n+a^2+1)\Big] H_n'(a)-16 (a^2-1)^2 H_n^3(a)\nonumber\\
&-&4 (a^2-1)\Big[a^4 (4 \alpha ^2+4 n^2+8 \alpha  n-1)+a^2 (8 \alpha ^2+4 n^2+8 \alpha  n+1)+4 n (2 \alpha +n)\Big]H_n^2(a)\nonumber\\
&-&16 a^2 (\alpha +n)^2  \Big[a^4 (2 \alpha ^2-n^2-2 \alpha  n-1)+3 a^2 n (2 \alpha +n)+1\Big]H_n(a)\nonumber\\
&+&32 a^6 (\alpha +n)^4 (a^2 n^2+2 \alpha  a^2 n+1)\bigg\}^2\nonumber\\
&=&64 (\alpha +n)^2 \Big[a^4 (\alpha +n)^2+(a^2-1) (H_n(a)-a H_n'(a))\Big]\bigg\{a^2 (a^2-1)^2 H_n''(a)\nonumber\\
&-&2 a (a^2-1) \Big[a^2 (2 n^2+4 \alpha  n+1)-2 H_n(a)\Big] H_n'(a)-4 (a^2-1) H_n^2(a)\nonumber\\
&-&2 a^2  \Big[(2 \alpha ^2-1) a^2+2 n (2 \alpha +n)+1\Big]H_n(a)+4 a^4 (\alpha +n)^2 (a^2 n^2+2 \alpha  a^2 n+1)\bigg\}^2.\label{hna}
\eea
\noindent{\bf Remark.}
Witte, Forrester and Cosgrove \cite{Witte} made use of the Fredholm theory of integral equations, following \cite{Tracy},
 to obtain a third order differential equation for a quantity related to the gap probability $\mathbb{P}(a,n)$, and in our paper this quantity
 is $\frac{H_n(a)}{2a}$. See Proposition 15 in \cite{Witte}.

Before we describe the choice of the double scaled variable appropriate to our problem. We give short description of the equilibrium density based on Dyson's Coulomb Fluid approach \cite{Dyson}

\begin{lemma}
The equilibrium density $\rho(x)$ of the eigenvalues (or particles) at the origin in the symmetric Jacobi unitary ensemble is $\frac{\sqrt{n(n+2\alpha)}}{\pi}$.
\end{lemma}
\begin{proof}
The density is characterized by the equilibrium equation \cite{Chen1997T},
\be\label{cf}
\mathrm{v}_{0}(x)-2\int_{-b}^{b}\ln|x-y|\rho(y)dy=A,
\ee
where $A$ is the Lagrange multiplier which fixes the constraint
\be\label{cf1}
\int_{-b}^{b}\rho(x)dx=n.
\ee
Taking a derivative with respect to $x$ from (\ref{cf}), we obtain a singular integral equation,
$$
\mathrm{v}_{0}'(x)-2P\int_{-b}^{b}\frac{\rho(y)}{x-y}dy=0.
$$
The solution subject to the boundary condition $\rho(-b)=\rho(b)=0$ is given by \cite{Gakhov}
$$
\rho(x)=\frac{\sqrt{b^2-x^2}}{2\pi^2}P\int_{-b}^{b}\frac{\mathrm{v}_{0}'(x)-\mathrm{v}_{0}'(y)}{x-y}\frac{dy}{\sqrt{b^2-y^2}}.
$$
Some simple computations show that
\be\label{rho}
\rho(x)=\frac{\alpha\sqrt{b^2-x^2}}{\pi\sqrt{1-b^2}(1-x^2)}.
\ee
Substituting (\ref{rho}) into (\ref{cf1}), we get
$$
b=\frac{\sqrt{n(n+2\alpha)}}{n+\alpha}.
$$
Hence,
$$
\rho(x)=\frac{\sqrt{n(n+2\alpha)-(n+\alpha)^2x^2}}{\pi(1-x^2)}.
$$
It follows that
$$
\rho(0)=\frac{\sqrt{n(n+2\alpha)}}{\pi}.
$$
\end{proof}
From Lemma 5.1, we see that the equilibrium density of the eigenvalues (or particles) at the origin is of order $n$ for large $n$, hence
for large $n$ and $a$ tends to $0$, we take the scaling variable to be $t:=c\:n\:a.$ Here $c$ is a constant to be suitably chosen later.
\\
Let
$$
\sigma_n(t):=-H_{n}\left(\frac{t}{cn}\right)
$$
\\
After this change of variables, (\ref{hna}) becomes,
\bea
&&\left[(ct\sigma_n''(t))^2+4c^2t(\sigma_n'(t))^3-(4c^2\sigma_n(t)-16t^2)(\sigma_n'(t))^2-32t\sigma_n(t)\sigma_n'(t)+16\sigma_n^2(t)\right]^2\nonumber\\
&+&\frac{64\alpha}{n}(\sigma_n(t)-t\sigma_n'(t))^2\left[(ct\sigma_n''(t))^2+4c^2t(\sigma_n'(t))^3-(4c^2\sigma_n(t)-16t^2)(\sigma_n'(t))^2-
32t\sigma_n(t)\sigma_n'(t)+16\sigma_n^2(t)\right]\nonumber\\
&+&O\left(\frac{1}{n^2}\right)=0.\nonumber
\eea
We have obtained the coefficients of the $O(n^{-2})$ term. However, since  as it is a rather large expression, we do not write it down here.
We observe first `portion' of the above equation has no explicit $n$ dependence, the $n$ dependence of the rest appear with coefficient $(1/n)$, $(1/n^2)$ and so on.

Now, letting $n\rightarrow\infty$, discarding the $(1/n)$ and $(1/n^2)$ and higher order terms, it is immediate that
$\sigma_n(t)$, replaced by $\sigma(t)$,  satisfies the following second order non-linear ODE.
$$
(t\sigma''(t))^2=-4t(\sigma'(t))^3+\frac{(4c^2 \sigma(t)-16t^2)(\sigma'(t))^2}{c^2}+\frac{32t\sigma(t)\sigma'(t)}{c^2}-\frac{16\sigma^2(t)}{c^2}.
$$
By choosing $c^2=-16$, namely $c=4\:i$, to fit into the original JMMS $\sigma$ form, we have
$$
(t\sigma''(t))^2=-4t(\sigma'(t))^3+(4\sigma(t)+t^2)(\sigma'(t))^2-2t\sigma(t)\sigma'(t)+\sigma^2(t).
$$
We summarize it in the following theorem.
\begin{theorem}
$\sigma(t)$ satisfies the Jimbo-Miwa-Okamoto $\sigma$ form of the Painlev\'{e} V equation \cite{Jimbo1981},
\bea
(t\sigma''(t))^2&=&\left[\sigma(t)-t\sigma '(t)+2(\sigma'(t))^2+(\nu_0+\nu_1+\nu_2+\nu_3)\sigma'(t)\right]^2\nonumber\\
&-&4(\nu_0+\sigma'(t))(\nu_1+\sigma'(t))(\nu_2+\sigma'(t))(\nu_3+\sigma'(t)),\nonumber
\eea
where
$$
\nu_0=\nu_1=\nu_2=\nu_3=0.
$$
\end{theorem}
The above is the JMMS equation governing the gap probability distribution of the Gaussian unitary ensemble \cite{Jimbo1980}. In \cite{Tracy}, a double scaling was carried out on the finite $n$ Gaussian Unitary Ensemble, (namely $w_0(x)=\mathrm{e}^{-x^2},$)  where $a:=t\:(2\sqrt{2n})^{-1}$, $n$ tends to $\infty$ and $t$ finite, one recovers the original JMMS $\sigma$ function. Here $2\sqrt{2n}$ is the density of eigenvalues of the GUE at the origin. In our problem the density of eigenvalues of the symmetric JUE, reads $\frac{\sqrt{n(n+2\alpha)}}{\pi}$
as demonstrated in the previous Lemma.

Our result is a demonstration of universality, namely,  that scaling at a small neighbourhood covering the
origin of the eigenvalue spectrum where the density is constant and large, the original JMMS $\sigma$ form is recovered.
\\
\noindent{\bf Remark.}
In JMMS, one also finds similar imaginary time transformation.
\section*{Acknowledgments}
The authors gratefully acknowledge the generous support of Macau Science and Technology Development Fund under the grant number FDCT 130/2014/A3, FDCT 023/2017/A1 and the University of Macau through MYRG 2014-00011-FST, MYRG 2014-00004-FST. We like to thank the referees for illuminating comments.

\end{document}